\documentclass[conference,letterpaper]{IEEEtran}

\usepackage{amsthm}
\usepackage[cmex10]{amsmath}
\usepackage{graphicx}
\usepackage[tight,footnotesize]{subfigure}
\usepackage{amsfonts}
\usepackage{algorithmic}
\usepackage{algorithm}
\usepackage{graphicx}
\usepackage{subfigure}
\usepackage{bm}
\usepackage{times,color,amssymb,epsfig,psfrag,stfloats,enumerate}
\usepackage{cite}
\usepackage{mathtools}
\usepackage{flushend}
\usepackage{booktabs}
\newtheorem{lemma}{Lemma}
\newtheorem{corollary}{Corollary}
\newtheorem{proposition}{Proposition}

\begin{document}
%
\title{Joint Transmission and Computing Scheduling for Status Update with Mobile Edge Computing}

\author{\IEEEauthorblockN{Jie Gong\IEEEauthorrefmark{1}, Qiaobin Kuang\IEEEauthorrefmark{2} and Xiang Chen\IEEEauthorrefmark{2}\\}
\IEEEauthorblockA{\IEEEauthorrefmark{1} School of Data and Computer Science, Sun Yat-sen University, Guangzhou 510006, China}
\IEEEauthorblockA{\IEEEauthorrefmark{2} School of Electronics and Information Technology, Sun Yat-sen University, Guangzhou 510006, China}
Email: gongj26@mail.sysu.edu.cn}


\maketitle

\begin{abstract}
Age of Information (AoI), defined as the time elapsed since the generation of the latest received update, is a promising performance metric to measure data freshness for real-time status monitoring. In many applications, status information needs to be extracted through computing, which can be processed at an edge server enabled by mobile edge computing (MEC). In this paper, we aim to minimize the average AoI within a given deadline by jointly scheduling the transmissions and computations of a series of update packets with deterministic transmission and computing times. The main analytical results are summarized as follows. Firstly, the minimum deadline to guarantee the successful transmission and computing of all packets is given. Secondly, a \emph{no-wait computing} policy which intuitively attains the minimum AoI is introduced, and the feasibility condition of the policy is derived. Finally, a closed-form optimal scheduling policy is obtained on the condition that the deadline exceeds a certain threshold. The behavior of the optimal transmission and computing policy is illustrated by numerical results with different values of the deadline, which validates the analytical results.
\end{abstract}


%
\IEEEpeerreviewmaketitle

\section{Introduction}
With the increasing demand of real-time status update applications such as autonomous driving, virtual reality and etc., age of information (AoI)  \cite{kaul2012real} is introduced as an effective data freshness metric, which is defined as the time elapsed since the generation of the latest received update. Recently, the impact of computing on AoI~\cite{alabbasi2018joint, arafa2019timely, zou2019trading, kuang2019age, gong2019reducing, zhong2019age, song2019age} is drawing more and more attention, as in many applications, the information embedded in a status update packet is not revealed until being processed. Due to the limited computing capacity of mobile devices, computing tasks for extracting information from status update packets are usually offloaded to the core network. As mobile edge computing (MEC) \cite{mao2017survey} can provide sufficient computing resources at the network edge, it is expected that the status update packets can be processed by an edge server while maintaining a low AoI. Since offloading a computing task includes transmission and computing, how to jointly optimize both procedures is a crucial problem.

The impact of computing on AoI was initially considered in~\cite{alabbasi2018joint}, where the computing tasks were scheduled in the central cloud. The scheduling policy for update cloud computing ignoring transmission time was studied in~\cite{arafa2019timely}. In~\cite{zou2019trading}, the tradeoff between computation and transmission was analyzed where each packet is pre-processed before being transmitted. When MEC is considered, the average AoI with exponential transmission time and service time was analyzed in~\cite{kuang2019age, gong2019reducing} for single user case. For multiple users, an optimal work-conserving scheduling policy was proposed in~\cite{zhong2019age}. A novel performance metric, age of task (AoT), was proposed in~\cite{song2019age} where task scheduling, computation offloading and energy consumption were jointly considered. Most of the existing computation related AoI analysis assumed a random computing time. However, in practice, computing time is usually fixed or predictable based on the volume of a task and the server capacity. When the data rate is fixed via some rate control mechanism, the transmission time is also fixed. This paper focuses on joint transmission and computing considering deterministic transmission time and computing time.

Since an MEC system can be viewed as a two-hop network, where the first hop is transmission and the second hop is computing, there were a lot of research efforts on multi-hop networks that can be referred for AoI analysis. In particular, Ref.~\cite{bedewy2017age} analyzed the optimality of the Last-Come-First-Serve (LCFS) queuing principle. The age-of-information for multi-flow multi-hop networks with interference was studied in~\cite{talak2017minimizing}. A useful tool named stochastic hybrid systems (SHS) was introduced in~\cite{yates2018age} showing the average AoI of a multi-hop line network with preemptive servers and Poisson arrivals. Nevertheless, these works still assumed random service time in each hop, and there still lack research efforts on deterministic transmission and computing times. A most closely related work studied the optimal offline scheduling policy in energy harvesting two-hop relay networks~\cite{ahmed2017age}. Different from relay networks where the two hops can not transmit simultaneously or may interfere with each other otherwise, transmission and computation can be scheduled at the same time.

In this paper, we study the average AoI minimization problem for transmitting and computing a set of packets before a given deadline. Each packet consumes a fixed transmission time and a fixed computing time. To minimize the average AoI, each packet should be transmitted upon its generation. The optimization variables include packet generation time instants and computation start time instants. We analyze the feasibility condition of the deadline to guarantee that all the packets can be successfully transmitted and computed. Then, a closed-form optimal scheduling policy is obtained when the deadline is sufficiently large. With moderate value of the deadline, the optimal solution can be found by the standard convex optimization algorithms. Numerical results illustrate the different behaviors of AoI curve with different deadlines.

\section{System Model} \label{sec:model}
As shown in Fig.~\ref{fig:system}, we consider an MEC-based status update system which is composed of a transceiver and an edge server. An update packet is generated by the source, and then transmitted through the channel to the remote edge server. The real-time status information embedded in the packet is exposed to the destination after being processed at the edge server. Assume the channel can only transmit a single packet at a time, and the edge server can only compute one packet at a time as well. As long as a packet arrives at the edge server, an acknowledgement will be sent back to the source, which can then generate a new packet. The packets arriving at the edge server can be computed if the server is idle. Otherwise, is has to be buffered to wait.

\begin{figure}[t]
\centering
\includegraphics[width=3.4in]{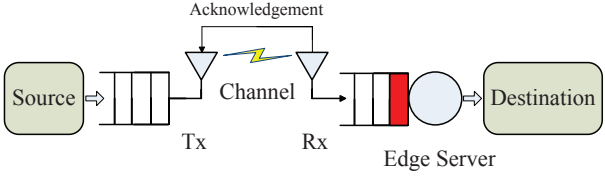}
\caption{Status update system with MEC.} \label{fig:system}
\end{figure}

We focus on the freshness of the status information, which is measured by the AoI, defined as
\begin{align}
\Delta(t) \coloneqq t - U(t),
\end{align}
where $U(t)$ is the generation time of the latest received packet at the destination. Denote
\begin{align}
\Delta_T \coloneqq \int_0^T \Delta(t)\mathrm d t.
\end{align}
Then, the average AoI for a given time period $T$ can be calculated as
\begin{align}
\bar{\Delta} \coloneqq \frac{1}{T}\Delta_T.
\end{align}

In this paper, we study a deterministic status update problem. In particular, there are a total of $N$ packets to be sent and processed within a given deadline $T$. The transmission time of the $k$-th packet is $T_k$ and its computing time is $C_k$. Denote $t_k$ as the time instant when the $k$-th packet is generated and transmitted. As only one packet can be transmitted at a time, we have
\begin{align}
t_k \ge t_{k-1} + T_{k-1}, \quad 2 \le k \le N.\nonumber
\end{align}
At the edge server, the $k$-th packet is processed from time instant $c_k$ after its reception. Since the edge server can only process a single packet at a time, we have
\begin{align}
c_k \ge c_{k-1} + C_{k-1}, \quad 2 \le k \le N.\nonumber
\end{align}
In addition, each packet can not be processed until it is received by the edge server, i.e.,
\begin{align}
c_k \ge t_k + T_k, \quad 1 \le k \le N.\nonumber
\end{align}
Finally, all the packets should be transmitted and processed before the deadline $T$, i.e.,
\begin{align}
c_N + C_N \le T.\nonumber
\end{align}

\begin{figure}[t]
\centering
\includegraphics[width=3.4in]{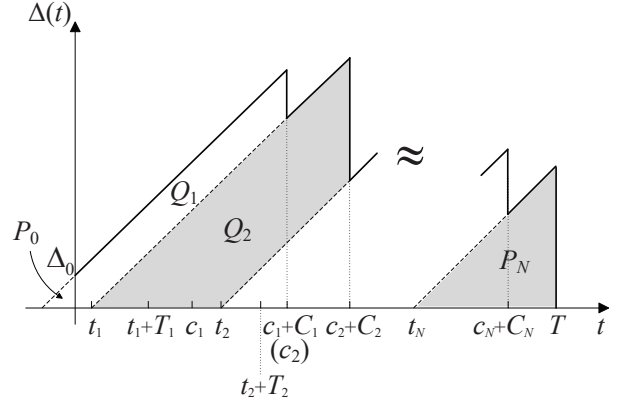}
\caption{A sample path of AoI with MEC.} \label{fig:aoi}
\end{figure}

In Fig.~\ref{fig:aoi}, a sample path of AoI with $N$ packets is depicted. We set the initial age as $\Delta(0) = \Delta_0$. The term $\Delta_T$ can be calculated as the area below the curve, which can be given by the summation of the areas of the trapezoids $Q_1, \cdots, Q_N$ plus the area of the triangle $P_N$ and minus the area of the triangle $P_0$, i.e.,
\begin{align}
\Delta_T = &\sum_{k=1}^N \frac{1}{2}\left[ (c_k+C_k-t_{k-1})^2 - (c_k+C_k-t_{k})^2\right] \nonumber\\
&\qquad + \frac{1}{2}(T-t_N)^2 - \frac{1}{2}\Delta_0^2, \label{eq:delta}
\end{align}
where we denote $t_0 = -\Delta_0$. Our objective is to minimize the average AoI by optimizing the packet transmission time instants $t_1, \cdots, t_N$ and the computing time instants $c_1, \cdots, c_N$. Since the time length $T$ is fixed, the problem can be equivalently formulated as follows:
\begin{subequations}\label{eq:prob}
\begin{align}
\min_{\mbox{\tiny$\begin{array}{c} t_1, \cdots, t_N \\ c_1, \cdots, c_N \end{array}$}} & \Delta_T\\
\mathrm{s.t.}\quad &t_1 \ge 0, \label{eq:t1}\\
&t_k \ge t_{k-1} + T_{k-1}, \quad 2 \le k \le N.\label{eq:tk}\\
&c_k \ge c_{k-1} + C_{k-1}, \quad 2 \le k \le N.\label{eq:ck}\\
&c_k \ge t_k + T_k, \quad 1 \le k \le N.\label{eq:ckt}\\
&c_N + C_N \le T. \label{eq:cn}
\end{align}
\end{subequations}
To guarantee that all the packets can be successfully transmitted and processed, the deadline $T$ needs to be long enough. The feasibility condition in terms of $T$ is given below.
\begin{proposition} \label{prop:feasi}
To transmit and compute $N$ packets before the deadline $T$, where each packet $k$ costs transmit time $T_k$ and computing time $C_k$, we have
\begin{align}
T \ge \max_{k \in \{1, \cdots, N \}} \left\{ \sum_{i=1}^k T_i + \sum_{j=k}^N C_j \right\}. \label{eq:Tcond}
\end{align}
\end{proposition}
\begin{proof}
According to \eqref{eq:ck} and \eqref{eq:ckt}, we have
\begin{align}
c_1 &\ge t_1 + T_1, \label{eq:c1max}\\
c_k &\ge \max\{c_{k-1} + C_{k-1}, t_k + T_k\}, \quad k\ge 2. \label{eq:ckmax}
\end{align}
Therefore, we have based on \eqref{eq:cn}
\begin{align}
T &\ge c_N + C_N \ge \max\{c_{N-1} + C_{N-1}, t_N + T_N\} + C_N \nonumber\\
&= \max\{c_{N-1} + C_{N-1} + C_N, t_N + T_N + C_N\}. \label{eq:T}
\end{align}
By applying \eqref{eq:tk} and \eqref{eq:ckmax} to \eqref{eq:T} recursively and with the boundary conditions \eqref{eq:c1max} and \eqref{eq:t1}, we can obtain \eqref{eq:Tcond}.
\end{proof}

In the case that \eqref{eq:Tcond} is satisfied with equality, the optimal solution for the problem \eqref{eq:prob} is trivial. In particular, the packets are greedily transmitted and computed as long as the channel or the edge server is idle. If the inequality in \eqref{eq:Tcond} strictly holds, the transmission and computing scheduling optimization is detailed in the next section.

\section{Optimal Transmission and Computing Policy}
In this section, we will solve the problem \eqref{eq:prob} to find the optimal time instants to transmit and compute the update packets for a feasible deadline $T$. Firstly, the monotonicity of the objective function is characterized in the following lemma.
\begin{lemma}\label{lemma:relation}
Under the constraints \eqref{eq:t1}-\eqref{eq:cn}, the objective $\Delta_T$ in the problem \eqref{eq:prob} is a non-increasing function of $t_k$ and an increasing function of $c_k$.
\end{lemma}
\begin{proof}
The lemma can be proved by directly taking the derivative of \eqref{eq:delta} with respect to $t_k$ and $c_k$. Specifically, as
\begin{align}
\frac{\partial \Delta_T}{\partial t_k} &= (c_k + C_k) - (c_{k+1} + C_{k+1}) < 0, 1 \le k \le N-1, \nonumber\\
\frac{\partial \Delta_T}{\partial t_N} &= (c_N + C_N) - T \le 0, \nonumber
\end{align}
$\Delta_T$ is a non-increasing function of $t_k$. As
\begin{align}
\frac{\partial \Delta_T}{\partial c_k} = t_k - t_{k-1} \ge T_{k-1} > 0,\quad 1 \le k \le N\nonumber
\end{align}
$\Delta_T$ is an increasing function of $c_k$.
\end{proof}

According to Lemma \ref{lemma:relation}, the optimal solution has the following relation.
\begin{corollary} \label{coro:opt}
The optimal solution for the problem \eqref{eq:prob} $t_1^*, \cdots, t_N^*, c_1^*, \cdots, c_N^*$ must satisfy
\begin{align}
c_1^* &= t_1^* + T_1, \label{eq:c1r}\\
c_k^* &= \max\{c_{k-1}^* + C_{k-1}, t_k^* + T_k\}, \quad 2 \le k \le N-1, \label{eq:ckr}\\
c_N^* &= t_N^* + T_N. \label{eq:cNr}
\end{align}
\end{corollary}
\begin{proof}
Since $\Delta_T$ is an increasing function of $c_k$, $c_1$ is lower bounded $t_1 + T_1$ according to \eqref{eq:ckt}, while $c_k$ is lower bounded by $\max\{c_{k-1} + C_{k-1}, t_k + T_k\}$ according to \eqref{eq:ck} and \eqref{eq:ckt} for $2 \le k \le N-1$. Hence, we have \eqref{eq:c1r} and \eqref{eq:ckr}.

Since $\Delta_T$ is a non-increasing function of $t_k$, and $t_N$ is upper bounded by $c_N - T_N$ according to \eqref{eq:ckt}, we have \eqref{eq:cNr}.
\end{proof}

With the monotonicity and Corollary \ref{coro:opt}, a greedy search algorithm can be developed. In particular, $t_k$ can be maximized by bisection search, and $c_k$ can be sequentially determined by $t_k$ and $c_{k-1}$. Nevertheless, the greedy search is time consuming especially when $N$ is large. We consider to develop more efficient algorithms in some special cases. Intuitively, to minimize the AoI, $t_k$ should be as large as possible while $c_k$ should be as small as possible. It is noticeable that the constraint \eqref{eq:ckt} plays an important role as it contains both an upper bound of $t_k$ and a lower bound of $c_k$. The following lemma shows when \eqref{eq:ckt} is satisfied with equality.

\begin{lemma} \label{lemma:equal}
Suppose the optimal solution for the problem \eqref{eq:prob} is $t_1^*, \cdots, t_N^*, c_1^*, \cdots, c_N^*$. If
\begin{align}
T \ge t_N^* + \max_{1 \le k \le N} \{c_k^* + C_k - t_k^*\}, \label{eq:Topt}
\end{align}
we have
\begin{align}
c_k^* = t_k^* + T_k, \quad 1 \le k \le N. \label{eq:ckstar}
\end{align}
\end{lemma}
\begin{proof}
According to \eqref{eq:ckr}, it is equivalent to prove $c_{k-1}^* + C_{k-1} \le t_k^* + T_k$. We prove it by contradiction. Suppose $c_{k-1}^* + C_{k-1} > t_k^* + T_k$ for a certain $k$. Denote
\begin{align}
\epsilon = \min\left\{(c_{k-1}^* + C_{k-1}) - (t_k^* + T_k),T-(c_N^* + C_N)\right\},\nonumber
\end{align}
and set a new solution for the problem \eqref{eq:prob} as follows
\begin{align}
&\tilde t_i = t_i^*, \tilde c_i = c_i^*, \quad i < k, \nonumber\\
&\tilde t_k = t_k^* + \epsilon, \tilde c_k = c_k^*,\nonumber\\
&\tilde t_i = t_i^* + \epsilon, \tilde c_i = c_i^* + \epsilon, \quad i > k.\nonumber
\end{align}
The above solution satisfies all the constraints \eqref{eq:t1}-\eqref{eq:cn}, and
\begin{align}
\tilde \Delta_T &= \sum_{k=1}^N \frac{1}{2}\left( (\tilde c_k+C_k-\tilde t_{k-1})^2 - (\tilde c_k+C_k-\tilde t_{k})^2\right) \nonumber\\
&\qquad + \frac{1}{2}(T-\tilde t_N)^2 - \frac{1}{2}\Delta_0^2\nonumber\\
&= \Delta_T^* - 2 \epsilon ((T - t_N^*) - (c_k^* + C_k - t_k^*)) \le \Delta_T^*\nonumber
\end{align}
according to \eqref{eq:Topt}. As a result, $\tilde t_1, \cdots, \tilde t_N, \tilde c_1, \cdots, \tilde c_N$ is optimal instead of $t_1^*, \cdots, t_N^*, c_1^*, \cdots, c_N^*$, which contradicts the assumption. Hence, we have $c_{k-1}^* + C_{k-1} \le t_k^* + T_k$ for all $k$, which results in \eqref{eq:ckstar}.
\end{proof}

The intuition behind Lemma \ref{lemma:equal} is that a packet waiting in the edge server will become stale. Therefore, each packet should arrive at the edge server right before the previous one completes its computing process, so that there is no waiting before computing. A policy satisfying \eqref{eq:ckstar} is termed as \emph{no-wait computing} policy. Notice that no-wait computing policy is not always optimal. When $T$ is relatively small, it maybe even not feasible. In Lemma \ref{lemma:equal}, how large value of $T$ is sufficient remains unsolved. In fact, it depends on the values of $T_k$s and $C_k$s. Firstly, we provide a feasibility condition for the no-wait computing policy.

\begin{proposition} \label{prop:Tfeasi}
No-wait computing policy satisfying \eqref{eq:ckstar} is feasible for the problem \eqref{eq:prob} if and only if
\begin{align}
T \ge T_1 + \sum_{k=1}^{N-1} \max\{C_k, T_{k+1}\} + C_N. \label{eq:Tfeasi}
\end{align}
\end{proposition}
\begin{proof}
Firstly, we prove the necessity. If the solution satisfies \eqref{eq:ckstar}, we have based on \eqref{eq:ck}
\begin{align}
t_k \ge t_{k-1} + T_{k-1} + C_{k-1} - T_k.\nonumber
\end{align}
Joint with \eqref{eq:tk}, we have
\begin{align}
t_k \ge t_{k-1} + T_{k-1} + \max\{C_{k-1}, T_k\} - T_k.\nonumber
\end{align}
Therefore, by recursively using the above inequality, we have
\begin{align}
T &\ge c_N + C_N \nonumber\\
&= t_N + T_N + C_N \nonumber\\
&\ge t_{N-1} + T_{N-1} + \max\{C_{N-1}, T_N\} + C_N \nonumber\\
&\ge t_1 + T_1 + \sum_{k=1}^{N-1} \max\{C_k, T_{k+1}\} + C_N \nonumber\\
&\ge T_1 + \sum_{k=1}^{N-1} \max\{C_k, T_{k+1}\} + C_N.\nonumber
\end{align}
Hence, the necessity of the condition \eqref{eq:Tfeasi} is proved.

Secondly, we prove the sufficiency by finding a feasible solution. In particular, we set
\begin{align}
\tilde t_1 &= 0, \nonumber\\
\tilde t_k &= \tilde t_{k-1} + T_{k-1} + \max\{0, C_{k-1} - T_k\}, \quad 2 \le k \le N, \nonumber\\
\tilde c_k &= \tilde t_k + T_k, \quad 1 \le k \le N. \nonumber
\end{align}
It is easy to verify that
\begin{align}
\tilde c_N + C_N  = T_1 + \sum_{k=1}^{N-1} \max\{C_k, T_{k+1}\} + C_N \le T.\nonumber
\end{align}
Hence, $\tilde t_1, \cdots, \tilde t_N, \tilde c_1, \cdots, \tilde c_N$ is a feasible solution of the problem \eqref{eq:prob} conditioned on \eqref{eq:Tfeasi}. As a result, the sufficiency is proved.
\end{proof}

Proposition \ref{prop:Tfeasi} provides a condition to guarantee that \eqref{eq:ckstar} ends up with a feasible solution. To guarantee its optimality, additional conditions are required, which is left for future work. Based on Lemma \ref{lemma:equal}, the optimal solution is no-wait computing policy if $T$ is sufficiently large. We will show the optimality numerically in the next section.

If the condition in Lemma \ref{lemma:equal} holds, the problem \eqref{eq:prob} can be simplified as follows
\begin{subequations} \label{eq:probsimp}
\begin{align}
\min_{t_1, \cdots, t_N} & \sum_{k=1}^N \frac{1}{2}\left[ (t_k-t_{k-1}+T_k + C_k)^2 - (T_k+C_k)^2\right] \nonumber\\
&\qquad + \frac{1}{2}(T-t_N)^2 - \frac{1}{2}\Delta_0^2\\
\mathrm{s.t.}\quad & t_1 \ge 0, \\
& t_k \ge t_{k-1} + T_{k-1}, \quad 2 \le k \le N, \label{eq:tk2}\\
& t_k + T_k \ge t_{k-1} + T_{k-1} + C_{k-1}, \quad 2 \le k \le N, \label{eq:tk3}\\
& t_N + T_N + C_N \le T.
\end{align}
\end{subequations}

By denoting
\begin{align}
x_k &= t_k - t_{k-1}+T_k + C_k, 1\le k \le N, \label{eq:xkk}\\
x_{N+1} &= T - t_N, \label{eq:xn1}
\end{align}
and changing the optimization variables, the problem \eqref{eq:probsimp} can be reformulated as
\begin{subequations}\label{eq:probchange}
\begin{align}
\min_{x_1, \cdots, x_{N+1}} & \sum_{k=1}^{N+1} \frac{1}{2}x_k^2 - \frac{1}{2}\sum_{k=1}^N(T_k+C_k)^2 - \frac{1}{2} \Delta_0^2 \label{eq:pcobj}\\
\mathrm{s.t.}\quad & x_1 \ge \Delta_0 + T_1 + C_1, \\
& x_k \ge T_{k-1} + T_k + C_k, \quad 2 \le k \le N, \label{eq:ak1}\\
& x_k \ge T_{k-1} + C_{k-1} + C_k, \quad 2 \le k \le N, \label{eq:ak2}\\
& x_{N+1} \ge T_N + C_N,\\
& \sum_{k=1}^{N+1} x_k = \Delta_0 + \sum_{k=1}^N(T_k+C_k) + T,\label{eq:sum}
\end{align}
\end{subequations}
where \eqref{eq:sum} is the constraint on the relation between the new variables. Notice that the second and third terms in \eqref{eq:pcobj} are constant which can be removed from the objective, and \eqref{eq:ak1} and \eqref{eq:ak2} can be merged together. By doing so, the optimal solution of the problem \eqref{eq:probchange} is equivalent to that of the following problem
\begin{subequations}\label{eq:pcsimp}
\begin{align}
\min_{x_1, \cdots, x_{N+1}} & \sum_{k=1}^{N+1} x_k^2 \label{eq:obj}\\
\mathrm{s.t.}\quad & x_1 \ge A_1, \label{eq:x1}\\
& x_k \ge A_k, \quad 2 \le k \le N, \label{eq:xk}\\
& x_{N+1} \ge A_N,\label{eq:xn}\\
& \sum_{k=1}^{N+1} x_k = B,\label{eq:sumx}
\end{align}
\end{subequations}
where
\begin{align}
A_1 &= \Delta_0 + T_1 + C_1, \label{eq:a1}\\
A_k &= T_{k-1} + \max\{C_{k-1}, T_k\} + C_k, \quad 2 \le k \le N, \\
A_{N+1} &= T_N + C_N, \\
B &= \Delta_0 + \sum_{k=1}^N(T_k+C_k) + T. \label{eq:b}
\end{align}
The solution to the problem \eqref{eq:pcsimp} can be geometrically interpreted in Fig.~\ref{fig:geo}. In particular, $\bm x = (x_1, \cdots, x_{N+1})^T$ can be viewed as a certain point in $\mathcal R^{N+1}$. The constraints \eqref{eq:x1}-\eqref{eq:xn} can be interpreted by the shaded region as shown in Fig.~\ref{fig:geo}, and the constraint \eqref{eq:sumx} can be interpreted by a hyperplane $l$. Hence, the feasible solution space is the intersection of the hyperplane $l$ and the shaded region, which is a polyhedron. The objective function is the square of the distance from a point $\bm x$ in the polyhedron to the origin. It is known that the minimum distance from the origin to the hyperplane is achieved by the vertical line between the origin and $\bm x^*$. Thus, if $\bm x^*$ lies in the shaded region, it is indeed the optimal solution of problem \eqref{eq:pcsimp}. The following proposition demonstrates that $\bm x^*$ is a feasible solution of the problem \eqref{eq:pcsimp} when $T$ is sufficiently large.

\begin{figure}[t]
\centering
\includegraphics[width=2.8in]{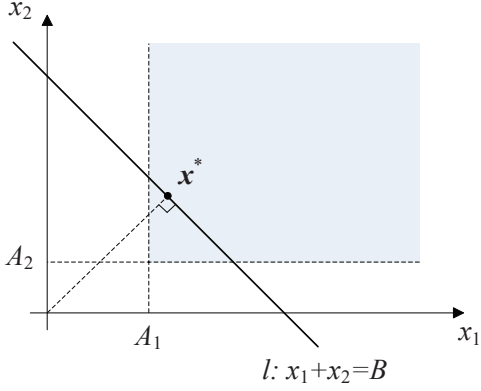}
\caption{Geometric interpretation of the solution for the problem \eqref{eq:pcsimp}.} \label{fig:geo}
\end{figure}

\begin{proposition} \label{thm:opt}
If
\begin{align}
T \ge (N+1) \max_{1\le k \le N+1}\{A_k\} - \sum_{k=1}^N(T_k + C_k) - \Delta_0, \label{eq:Tge}
\end{align}
the optimal solution of the problem \eqref{eq:pcsimp} is $\bm x^* = \left(x_1^*, \cdots, x_{N+1}^*\right)^T$ with
\begin{align}
x_k^* = \frac{B}{N+1}, \quad 1\le k \le N+1,
\end{align}
where $A_k$ and $B$ are defined as \eqref{eq:a1}-\eqref{eq:b}.
\end{proposition}
\begin{proof}
If \eqref{eq:Tge} holds, we have
\begin{align}
x_k^* = \frac{B}{N+1} \ge \max_{1\le k \le N+1} \{A_k\}. \nonumber
\end{align}
Therefore, all the constraints \eqref{eq:x1}, \eqref{eq:xk} and \eqref{eq:xn} are satisfied. Obviously, the constraint \eqref{eq:b} is also satisfied. Hence, $\bm x^*$ is a feasible solution of the problem \eqref{eq:pcsimp}. Since $\bm x^*$ is the intersection point of the hyperplane $l$ and its vertical line from the origin, it attains the minimum distance as in \eqref{eq:pcobj}.
\end{proof}
As the value of $T$ satisfying \eqref{eq:Tge} is sufficiently large, we can infer that the relation \eqref{eq:ckstar} holds for the optimal solution. Consequently, the optimal solution for the original problem \eqref{eq:prob} can be directly obtained according to \eqref{eq:xkk} and \eqref{eq:xn1}, which is given as follows,
\begin{align}
t_k^* &= \frac{k}{N+1}B-\sum_{i=1}^k(T_i + C_i) - \Delta_0, \label{eq:tkstar}\\
c_k^* &= \frac{k}{N+1}B-\sum_{i=1}^{k-1}(T_i + C_i) - C_k - \Delta_0. \label{eq:ckstarr}
\end{align}

Proposition \ref{thm:opt} can be explained as follows. Notice that
\begin{align}
x_k^* &= t_k^* - t_{k-1}^* + T_k + C_k \nonumber\\
&= c_k^* + C_k - t_{k-1}^*\nonumber
\end{align}
is the local maximum value of the AoI curve as shown in Fig.~\ref{fig:aoi}, which is referred to as the peak AoI \cite{costa2014age}. Since all the peak AoIs are the same for all the packets, all the isosceles right-angled triangles with side length $c_k^* + C_k - t_{k-1}^*$ in Fig.~\ref{fig:aoi} are of the same size. Therefore, the average AoI is minimized when the ``contributions" of all the packets are the same.

In the case that $\bm x^*$ does not lie in the shaded region, the optimal solution can be found on the boundary of the hyperspace. It is obvious that the problem \eqref{eq:pcsimp} is a convex optimization problem as the objective is a quadratic function and the constraints are all linear. Therefore, it can be solved by the standard convex optimization algorithms \cite{boyd2004convex}.

\section{Numerical Results}

In this section, we show the behavior of the optimal solution by numerical studies. In the numerical experiment, we set the number of packets $N = 5$, and the set of transmission times as $(T_1, T_2, T_3, T_4, T_5) = (0.5\mathrm{s} ,0.1\mathrm{s} ,0.3\mathrm{s}, 0.7\mathrm{s} ,0.4\mathrm{s})$, the set of the computing times as $(C_1 ,C_2 ,C_3, C_4 ,C_5) = (0.2\mathrm{s}, 0.4\mathrm{s}, 0.3\mathrm{s}, 0.6\mathrm{s} , 0.8\mathrm{s})$ according to \cite{Chen2016Efficient}. The initial age is set to $\Delta_0 = 1\mathrm{s}$. By choosing different values of $T$, the optimal solutions are different. In particular, we solve the AoI minimization problem numerically for $T = 3\mathrm{s}, 5\mathrm{s}$, and $7.5\mathrm{s}$, respectively. The AoI curves and the optimal scheduling solutions are shown in Figs.~\ref{fig:subfig_evo:Tmin}, \ref{fig:subfig_evo:Tmid} and \ref{fig:subfig_evo:Tlar}.

\begin{figure}[t]
    \includegraphics[width=3.4in]{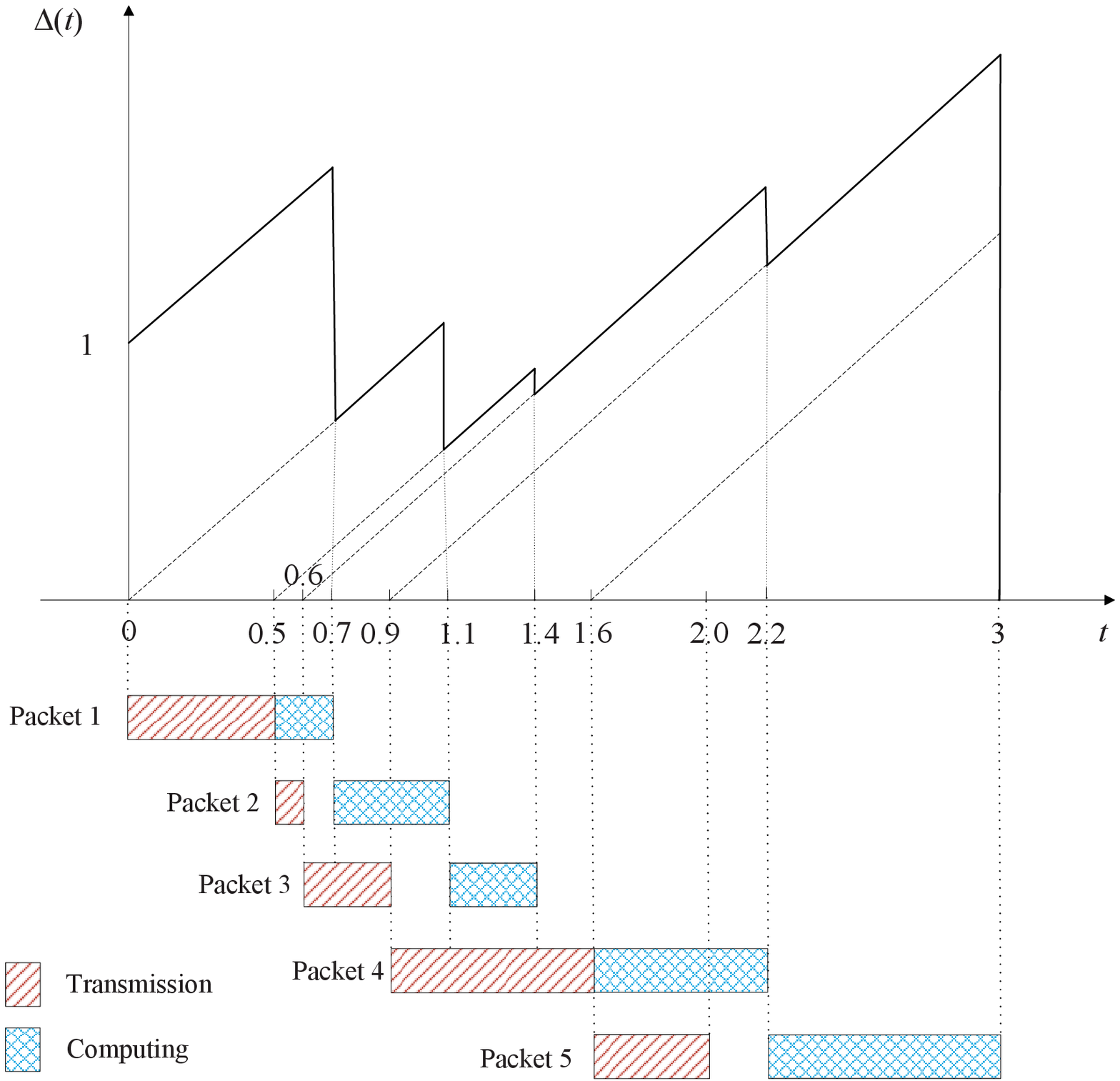}
    \caption{AoI curve $\Delta(t)$ and optimal scheduling policy with $T =3\mathrm{s}$.}
    \label{fig:subfig_evo:Tmin} 
\end{figure}

\begin{figure}[t]
    \includegraphics[width=3.4in]{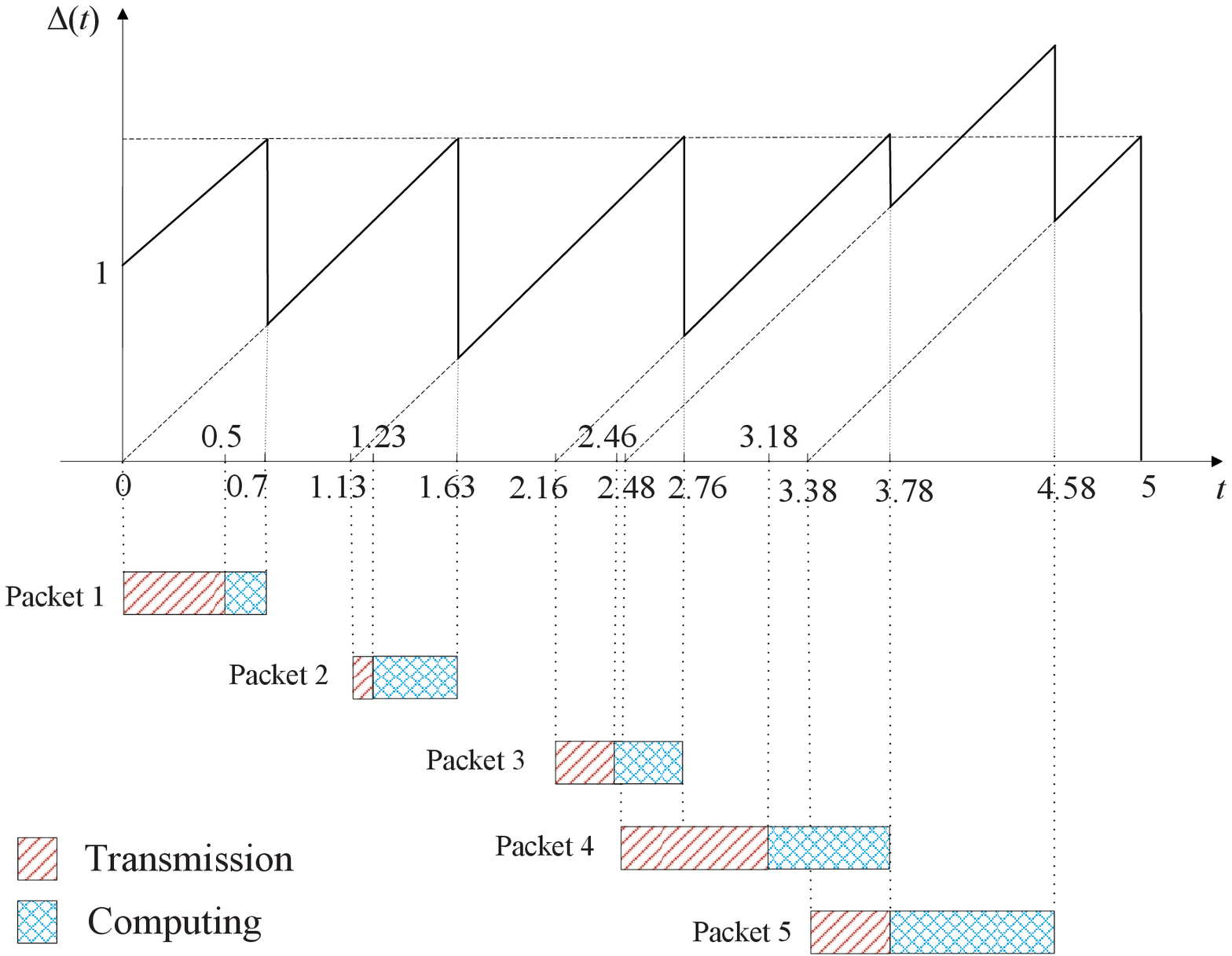}
    \caption{AoI curve $\Delta(t)$ and optimal scheduling policy with $T =5\mathrm{s}$.}
    \label{fig:subfig_evo:Tmid} 
\end{figure}

\begin{figure}[t]
    \includegraphics[width=3.4in]{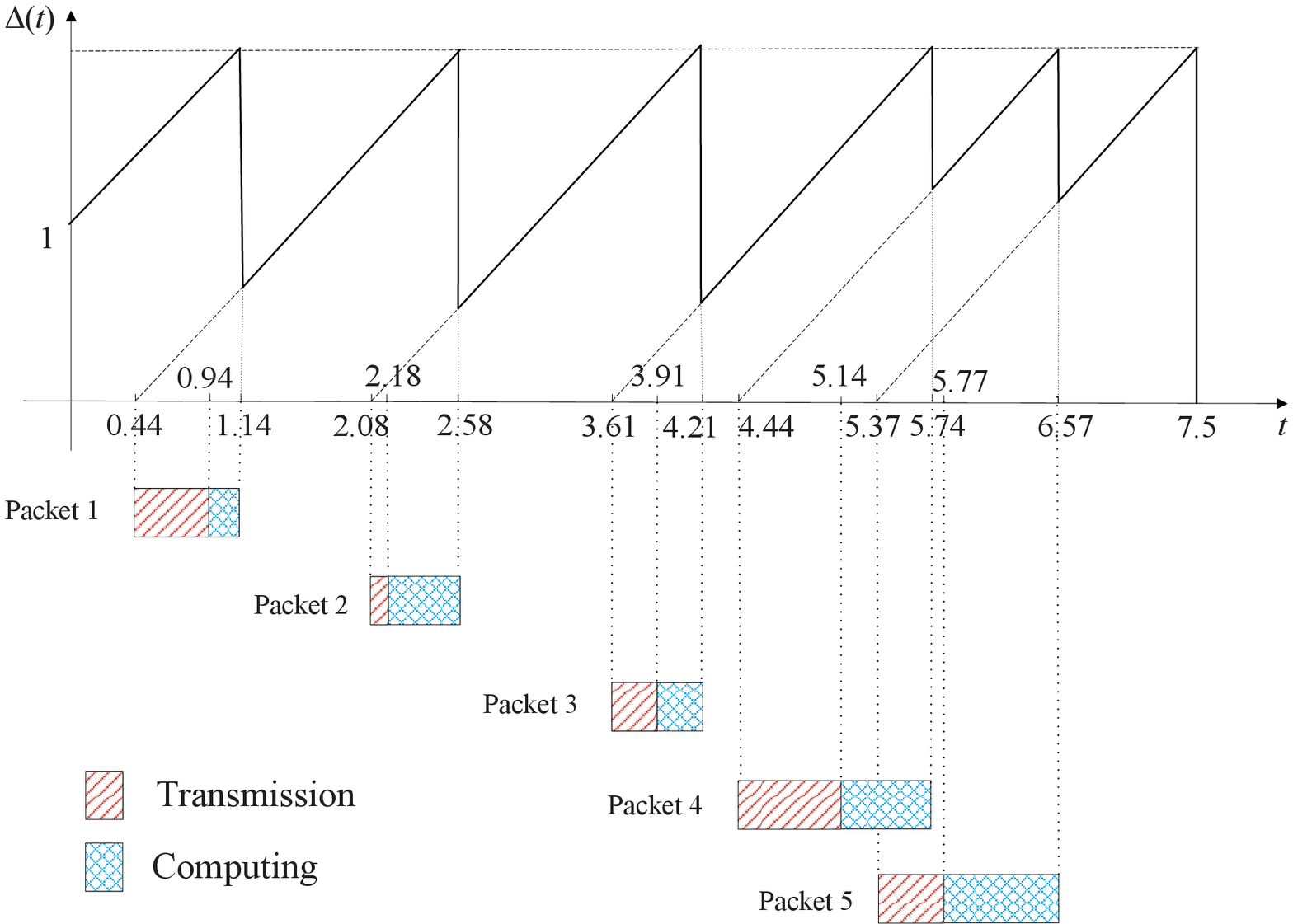}
    \caption{AoI curve $\Delta(t)$ and optimal scheduling policy with $T =7.5\mathrm{s}$.}
    \label{fig:subfig_evo:Tlar} 
\end{figure}

Firstly, according to Proposition \ref{prop:feasi}, to make sure all the packets can be successfully transmitted and computed, we have $T \ge 3\mathrm{s}$. Therefore, when $T=3\mathrm{s}$, an optimal solution is to transmit and compute the packets greedily, i.e., each packet is generated and transmitted upon receiving the previous packet, and each packet is computed as long as the edge server is idle as shown in Fig.~\ref{fig:subfig_evo:Tmin}. It is noticeable that $c_k = t_k + T_k$ does not always hold due to the limited time length. It is also remarkable that there may exist other feasible optimal solutions in this case. As shown in this figure, the fifth packet can start to transmit during time interval $(1.6\mathrm{s}, 1.8\mathrm{s})$ without changing the average AoI.

Secondly, when $T = 5\mathrm{s}$, the optimal solution is depicted in Fig.~\ref{fig:subfig_evo:Tmid}. It can be seen that $c_k = t_k + T_k$ holds for all $k$. Therefore, $T = 5\mathrm{s}$ is sufficiently large so that the optimal solution is no-wait computing policy. Finally, in the case that $T=7.5\mathrm{s}$, the condition \eqref{eq:Tge} in Proposition \ref{thm:opt} is satisfied. Hence, the optimal solution is solved in closed-form as \eqref{eq:tkstar} and \eqref{eq:ckstarr}. The numerical result in Fig.~\ref{fig:subfig_evo:Tlar} validates our theoretical analysis.


By observing the AoI curves with different $T$s depicted in Figs.~\ref{fig:subfig_evo:Tmin}, \ref{fig:subfig_evo:Tmid} and \ref{fig:subfig_evo:Tlar}, it can be seen that when $T =3\mathrm{s}$, the peak AoI for each packet varies with one another. As $T$ increases, the variance among the peak AoIs becomes small. In particular, when $T=5\mathrm{s}$, all the peak AoIs are the same except for the fourth packet. And when $T = 7.5\mathrm{s}$, the peak AoIs are all the same, which is consistent with the explanation of Proposition \ref{thm:opt}. It is demonstrated that the more ``flat" the AoI curve is, the smaller the average AoI is.

%
%

\section{Conclusion}
In this paper, we have studied how to schedule the transmission and computing of a set of packets in a serial way to minimize the average AoI. The optimal solution strongly depends on the time deadline. If the deadline is just enough for completing the transmission and computing of all the packets, there is no space for adjustment. If the deadline is sufficiently large, all the peak AoIs are the same so that the average AoI is minimized, and the optimal solution can be given in closed-form. For moderate length of the deadline, to minimize the average AoI, each packet should start computing right after it is received by the edge server so that no waiting occurs in the buffer of the edge server, referred to as no-wait computing policy. Future work may include finding necessary and sufficient condition to guarantee that no-wait computing policy is optimal, and characterizing the optimal solution on the boundary of the hyperplane.

\bibliographystyle{IEEEtran}
\bibliography{ref}

\end{document}